\documentclass[11pt]{article}

\usepackage[margin=1in]{geometry}
\usepackage{graphicx}
\usepackage[pdftex,colorlinks=true,linkcolor=blue,citecolor=blue,urlcolor=black]{hyperref}
\usepackage{amsmath, amsthm, amssymb}
\usepackage{subfigure}
\usepackage{comment}
\usepackage{soul}
\usepackage{url}
\usepackage{pdflscape}
\usepackage[ruled,lined,linesnumbered]{algorithm2e}
\usepackage{tikz}

\newcommand{\R}{\mathbb{R}}

\newcommand{\C}{\mathbb{C}}

\newcommand{\E}{\mathbb{E}}

\newcommand{\ket}[1]{| #1 \rangle}
\newcommand{\bra}[1]{\langle #1|}

\newcommand{\hsip}[2]{\langle #1,#2 \rangle}
\newcommand{\proj}[1]{| #1 \rangle \langle #1 |}

\DeclareMathOperator{\tr}{tr}

\DeclareMathOperator{\spann}{span}

\DeclareMathOperator{\Ent}{Ent}

\newcommand{\be}{\begin{equation}}
\newcommand{\ee}{\end{equation}}
\newcommand{\bea}{\begin{eqnarray}}
\newcommand{\eea}{\end{eqnarray}}
\newcommand{\bes}{\begin{equation*}}
\newcommand{\ees}{\end{equation*}}
\newcommand{\beas}{\begin{eqnarray*}}
\newcommand{\eeas}{\end{eqnarray*}}

\makeatletter
\newtheorem*{rep@theorem}{\rep@title}
\newcommand{\newreptheorem}[2]{%
\newenvironment{rep#1}[1]{%
 \def\rep@title{#2 \ref{##1} (restated)}%
 \begin{rep@theorem}}%
 {\end{rep@theorem}}}
\makeatother

\newtheorem{thm}{Theorem}
\newtheorem*{thm*}{Theorem}
\newtheorem{cor}[thm]{Corollary}

\newtheorem{lem}[thm]{Lemma}

\newtheorem*{lem*}{Lemma}

\newreptheorem{thm}{Theorem}
\newreptheorem{lem}{Lemma}

\begin{document}

\title{Quantum reverse hypercontractivity}
\author{Toby Cubitt$^{1,2}$, Michael Kastoryano$^3$,
Ashley Montanaro$^4$ and Kristan Temme$^5$\\[8pt]
{\small $^1$ Department of Computer Science, University College London, UK} \\
{\small $^2$ Centre for Quantum Information and Foundations, DAMTP, University of Cambridge, UK} \\
{\small $^3$ NBIA, Niels Bohr Institute, University of Copenhagen, 2100 Copenhagen, DK}\\
{\small $^4$ School of Mathematics, University of Bristol, UK}\\
{\small $^5$ Institute for Quantum Information and Matter, California Institute of Technology,}\\
{\small Pasadena CA 91125, USA}
}

\date{\today}

\maketitle

\begin{abstract}
We develop reverse versions of hypercontractive inequalities for quantum channels. By generalizing classical techniques, we prove a reverse hypercontractive inequality for tensor products of qubit depolarizing channels. We apply this to obtain a rapid mixing result for depolarizing noise applied to large subspaces, and to prove bounds on a quantum generalization of non-interactive correlation distillation.
\end{abstract}

% ------------------------------------------------------------------------------

\section{Introduction}

The theory of hypercontractivity has become an essential tool in disciplines ranging from quantum field theory~\cite{gross06}, to theoretical computer science~\cite{dewolf08}, to quantum information theory~\cite{montanaro12}. One of the simplest and most well-known, yet also most important, results in this area is hypercontractivity of a certain noise operator acting on the boolean cube $\{0,1\}^n$. The noise operator $T_\gamma$ is defined by $(T_\gamma f)(x) = \E_{y \sim_\epsilon x}[f(y)]$ for functions $f:\{0,1\}^n \rightarrow \R$, where $y$ is distributed such that each bit of $y$ is equal to the corresponding bit of $x$, except with independent probability $\epsilon = (1-\gamma)/2$. Let $\|f\|_p$ be the normalized $\ell_p$ norm, $\|f\|_p = \left(\frac{1}{2^n}\sum_{x \in \{0,1\}^n} |f(x)|^p \right)^{1/p}$. Then it is easy to show that, for any $p \ge 1$ and any $\gamma \in [0,1]$, $T_\gamma$ is a contraction: $\|T_\gamma (f) \|_p \le \|f\|_p$.

However, a stronger result holds:

\begin{thm}[Hypercontractivity for the boolean cube~\cite{bonami70,gross75,beckner75}]
\label{thm:forwardclassical}
Let $f:\{0,1\}^n \rightarrow \R$ and fix $p$ and $q$ such that $1 \le p \le q \le \infty$. Then, for all $\gamma$ such that $0 \le \gamma \le \sqrt{(p-1)/(q-1)}$,
\[ \| T_\gamma(f) \|_q \le \|f\|_p. \]
\end{thm}

Theorem \ref{thm:forwardclassical} is the key technical result in applications to computer science such as the famous result of Kahn, Kalai and Linial that every boolean function has an influential variable~\cite{kahn88}. One can generalize Theorem \ref{thm:forwardclassical} to a quantum (i.e.\ noncommutative) setting. Here the natural generalization of the noise operator $T_\gamma$ acting on $n$ bits is the qubit depolarizing channel $\mathcal{D}_\gamma$ acting on $n$ qubits, applied to each qubit independently. The qubit depolarizing channel is defined by
\[ \mathcal{D}_\gamma(f) = (1-\gamma) (\tr f) \frac{I}{2} + \gamma f, \]
where $f \in \mathcal{M}_2$ is a linear operator acting on the space of one qubit. The failure of multiplicativity for maximum output $p$-norms of quantum channels implies that it is more challenging to study tensor products of quantum channels than in the classical case. However, the following generalization of Theorem \ref{thm:forwardclassical} was proven by King~\cite{king14}, generalizing a previous result of Montanaro and Osborne~\cite{montanaro10c}. Related special cases were previously shown by Carlen and Lieb~\cite{carlen93} and Biane~\cite{biane97a}; see~\cite{king14} for a discussion.

\begin{thm}[Quantum hypercontractivity for the depolarizing channel~\cite{king14}]
\label{thm:forwardquantum}
Let $f \in \mathcal{M}_{2^n}$ and fix $p$ and $q$ such that $1 \le p \le q \le \infty$. Then, for all $\gamma$ such that $0 \le \gamma \le \sqrt{(p-1)/(q-1)}$,
\[ \| \mathcal{D}_\gamma^{\otimes n}(f) \|_q \le \|f\|_p. \]
\end{thm}

In this theorem $\|\cdot\|_p$ is the normalized $\mathbb{L}_p$ norm (Schatten $p$-norm), $\|f\|_p = \left(\frac{1}{2^n} \tr |f|^p \right)^{1/p}$, where $|f| = \sqrt{f^\dag f}$. Theorem \ref{thm:forwardquantum} has found its own applications, to spectral bounds for local Hamiltonians~\cite{montanaro12} and bounds on mixing times for the depolarizing channel~\cite{kastoryano13}.

One can also define the $\ell_p$ and $\mathbb{L}_p$ ``norms'' for any $p < 1$, though in this case these functions are no longer actually norms. A result similar in appearance to Theorem \ref{thm:forwardclassical}, but perhaps less well-known, was proven by Borell~\cite{borell82} for the case where $p$ and $q$ are less than 1:
\begin{thm}[Reverse hypercontractivity for the boolean cube~\cite{borell82}]
\label{thm:reverseclassical}
Let $f:\{0,1\}^n \rightarrow \R$ be non-negative and fix $p$ and $q$ such that $-\infty < q \le p \le 1$ (if $p < 0$, assume that $f$ is strictly positive). Then, for all $\gamma$ such that $0 \le \gamma \le \sqrt{(1-p)/(1-q)}$,
\[ \| T_\gamma(f) \|_q \ge \|f\|_p. \]
\end{thm}
Observe that the inequality in Theorem \ref{thm:reverseclassical} is reversed as compared with Theorem \ref{thm:forwardclassical}. Theorem \ref{thm:reverseclassical} has also found applications in classical computer science, such as isoperimetric inequalities, bounds on correlation distillation and quantitative versions of Arrow's theorem~\cite{mossel06,mossel11,mossel13a}.

Here we generalize Theorem \ref{thm:reverseclassical} to the quantum setting, proving the following result:
\begin{thm}
\label{thm:reversequantumdep}
Let $f \in \mathcal{M}_{2^n}$ be positive semidefinite and let $-\infty < q \le p \le 1$ (if $p < 0$, assume that $f$ is positive definite). Then, for all $\gamma$ such that $0 \le \gamma \le \sqrt{(1-p)/(1-q)}$,
\[ \| \mathcal{D}_\gamma^{\otimes n}(f) \|_q \ge \|f\|_p. \]
\end{thm}

We apply Theorem \ref{thm:reversequantumdep} to obtain quantum equivalents of some of the main classical applications of reverse hypercontractivity~\cite{mossel06,mossel13a}. First, we obtain a ``rapid mixing'' result for large subspaces. Imagine that $S$ and $T$ are large subspaces of the space of $n$ qubits. One special case of our result is that, if we take the uniform mixture on $S$ and apply a certain amount of depolarizing noise to each qubit, the resulting state will have quite a large overlap with $T$. Second, we study a quantum generalization of the problem of non-interactive correlation distillation~\cite{mossel05,mossel06,mossel13a}. Classically, this is a game with $k$ players, where a copy of a string of $n$ random bits is distributed to each player, with independent noise applied to each copy. The goal of the players is to output a single, shared, uniformly random bit. Here we generalize this by replacing the random bits with a random pure state of $n$ qubits picked from a non-product basis, and put limits on the success probability of the players in this more general setting.

% ------------------------------------------------------------------------------

\section{Preliminaries}

We let $\mathcal{M}_d$ denote the set of complex $d \times d$ matrices. We write $\mathcal{T}_t: \mathcal{M}_d \rightarrow \mathcal{M}_d$ for a general Markovian family of quantum channels generated by a Lindblad operator $\mathcal{L}$:
\[ \mathcal{T}_t = e^{-t \mathcal{L}}. \]
The important special case of the qubit depolarizing channel fits into this picture. We can write $\mathcal{D}_\gamma = e^{-t \mathcal{L}}$, where $\gamma = e^{-t}$ and $\mathcal{L}(f) = f - (\tr f) \frac{I}{2}$.

Let $\tau$ be the normalised trace, $\tau(f) = \tr f / d$ for $f \in \mathcal{M}_d$. Further define $\hsip{f}{g} = \tau(f^\dag g)$. The entropy of $f$ is defined as
\[ \text{Ent}(f) = \tau(f \ln f) - \tau(f)\ln\tau( f). \]
%
%Let $\mathcal{L}^{(n)} = \sum_i \mathcal{L}_i$ denote the sum of the operators $\mathcal{L}$ applied on each of the $n$ qubits. The Dirichlet form associated with $\mathcal{L}^{(n)}$ is
For $\mathcal{L}$ the generator of a semigroup, the associated Dirichlet form is
\[ \mathcal{E}_\mathcal{L}(f,g) = \tau(f \mathcal{L}(g)). \]
For $p \in \R$, the H\"older conjugate $p'$ is defined by $1/p + 1/p' = 1$.

A quantum channel (completely positive, trace-preserving map) $\mathcal{T}:\mathcal{M}_d \rightarrow \mathcal{M}_d$ is said to be unital if $\mathcal{T}(I)= I$. $\mathcal{T}$ is said to be primitive if for all density matrices $\rho$ the limit $\lim_{k \rightarrow \infty} \mathcal{T}^k(\rho)$ exists and is equal to some $\sigma>0$ independent of $\rho$. A unital channel $\mathcal{T}$ is said to be \textit{reversible} if $\tau(f \mathcal{T}(g))=\tau(\mathcal{T}(f) g)$ for all $f,g \in \mathcal{M}_d$, i.e.\ the map is Hermitian with respect to the Hilbert-Schmidt inner product; equivalently $\mathcal{T}$'s Kraus operators are Hermitian. (Note that ``reversible'' is the standard terminology for this property in classical Markov chains, which in our context is the special case in which $\mathcal{T}$ is a classical channel (stochastic map). See e.g.\ \cite{temme10} and references therein for a more extensive discussion of generalisations of Markov chain properties to the quantum setting.)

For unital reversible channels, hypercontractivity of $\mathcal{T}$ (Schr\"odinger picture) and $\mathcal{T}^\dag$ (Heisenberg picture) are equivalent, although this is not true in general. Since we are only concerned with unital reversible channels in this paper, we work in the Schr\"odinger picture throughout. 
% ------------------------------------------------------------------------------

\subsection{Reverse norm inequalities}

We first give reverse versions of various standard $\mathbb{L}_p$ norm inequalities. These inequalities seem likely to be well-known, but we include proofs for completeness where these are not easy to find in the literature.

First we state the reverse H\"older inequality for operators, a proof of which can be found in Ref.~\cite{tomamichel14}:

\begin{lem}
\label{lem:revholder}
Let $f\geq0$ be a positive semidefinite operator and let $g>0$ be a positive definite operator. Fix $p$ and $p'$ such that $0 \le p \le 1$ and  $1/p + 1/p' = 1$. Then
\[ \tau(fg) \ge \|f\|_p \|g\|_{p'}. \]
\end{lem}

We will also need the reverse Minkowski inequality:

\begin{lem}
\label{lem:revmink}
Let $f$ and $g$ be positive semidefinite operators and let $p < 1$ (if $p < 0$, $f$ and $g$ need to be positive definite). Then
\[ \|f + g \|_p \ge \|f\|_p + \|g\|_p. \]
\end{lem}

\begin{proof}
The proof can be obtained from the reverse H\" older inequality. Let $f,g$ be nonzero positive semidefinite operators. Then
\beas \|f+g\|_p^p&=&\|(f+g)(f+g)^{p-1}\|_1\\
&=&\|f(f+g)^{p-1}\|_1+\|g(f+g)^{p-1}\|_1\\
&\geq&(\|f\|_p+\|g\|_p)\frac{\|f+g\|_p^p}{\|f+g\|_p}\eeas
where the inequality follows from the reverse H\" older inequality. The claimed inequality follows by rearranging.
\end{proof}

The following lemma gives a variational characterization for $p<1$:

\begin{lem}
\label{lem:varchar}
Let $f$ be positive semidefinite operators and let $p < 1$ (if $p < 0$, $f$ needs to be positive definite). Then
\[ \|f\|_p=\inf\{\tau(fg): g>0, \|g\|_{p'}\geq 1\}. \]
\end{lem}
\begin{proof}
The proof again follows directly from reverse H\"older.
\end{proof}

Finally, we show that the depolarizing channel is expansive for $p < 1$:

\begin{lem}
\label{lem:expansion}
Let $f$ be a positive semidefinite operator on $\mathcal{M}_d$ and let $p < 1$ (if $p < 0$, $f$ needs to be positive definite). Let $\mathcal{T}:\mathcal{M}_d\rightarrow\mathcal{M}_d$ be a unital quantum channel. Then
\[ \|\mathcal{T}(f) \|_p \ge \|f\|_p. \]
\end{lem}

\begin{proof}
It is sufficient to show that $\tr |\mathcal{T}(f)|^p \ge \tr |f|^p$ for $p \ge 0$, and $\tr |\mathcal{T}(f)|^p \le \tr |f|^p$ for $p < 0$. We write the channel in Kraus form, $\mathcal{T}(f)=\sum_\alpha A_\alpha f A_\alpha^\dag$ for some Kraus operators $\{A_\alpha\}$, and use the following operator Jensen's inequality (see e.g.~\cite{hansen03}): If $\phi$ is a convex function, then
\[ \tr \phi\left(\sum_\alpha B_\alpha^\dag f B_\alpha\right)\leq \tr \sum_\alpha B_\alpha^\dag \phi(f) B_\alpha \]
for any positive semidefinite operator $f$, and an arbitrary sequence of operators $(B_\alpha)$ such that $\sum_\alpha B_\alpha^\dag B_\alpha = I$. For concave functions the inequality is reversed. As $\mathcal{T}$ is unital we have $\sum_\alpha A_\alpha A_\alpha^\dag = I$, so taking $B_\alpha = A_\alpha^\dag$ and choosing $\phi(x)=|x|^p$ (which is concave for $p\ge 0$ and convex for $p\le0$) yields expansivity of the channel.
\end{proof}

% ------------------------------------------------------------------------------

\section{Proof of Theorem \ref{thm:reversequantumdep}}

Our proof of Theorem \ref{thm:reversequantumdep} will be based on putting together a sequence of small, and somewhat more general lemmas; Theorem \ref{thm:reversequantumdep} then becomes a corollary. The overall strategy is to generalize to the quantum setting the proof technique used by Mossel, Oleszkiewicz and Sen~\cite{mossel13a} to prove general reverse hypercontractive inequalities classically.  The main technical tool we will need is a quantum generalization of a classical inequality proven by these authors~\cite{mossel13a}, which in turn generalizes an inequality of Stroock~\cite{stroock84} and Varopoulos~\cite{varopoulos85} to $p,q < 1$.

\begin{lem}[Quantum Stroock-Varopoulos inequality]
\label{lem:qstroock}
Let $\mathcal{T}$ be a unital, reversible and primitive quantum channel. Let ${\cal L} = c_0\left({\bf id} - \mathcal{T}\right)$, with $c_0 \in \R_+$, denote the generator of a reversible quantum semigroup. Then for $p \geq q$ with $p,q \in (0,2] \backslash \{1\}$ and positive semidefinite $g \in {\cal M}_d$ we have
\[
  p p' {\cal E}_\mathcal{L}(g^{1/p'},g^{1/p}) \leq q q' {\cal E}_\mathcal{L}(g^{1/q'},g^{1/q}).
\]
Here $q'$ and $p'$ denote the H{\"o}lder conjugates of $q$ and $p$ respectively.
\end{lem}

\begin{proof}
We prove the lemma by showing that the required inequality follows from the classical generalized Stroock-Varopoulos inequality~\cite{mossel13a} when considering reversible generators of the form ${\cal L} = c_0\left({\bf id} - \mathcal{T}\right)$.  Consider for $a,b \in \R$ the matrix power $g^a = \sum_i g_i^a \proj{i}$ in the eigenbasis of $g$. Then
\[
{\cal E}_\mathcal{L}(g^a,g^b)  = \tau\left( g^a {\cal L}(g^b)\right) = \frac{1}{d} \sum_{i,j} \; g_i^a g_j^b \; \bra{i} {\cal L}(\proj{j}) \ket{i}.
\]
Given the Kraus decomposition $\mathcal{T}(f) = \sum_\alpha A_\alpha f A_\alpha^\dag$, we have $\bra{i} {\cal L}(\proj{j}) \ket{i} = c_0(\delta_{ij} - P_{ij})$, where $P_{ij} = \sum_\alpha
|\bra{i}A_\alpha\ket{j}|^2 \geq 0$ is a doubly stochastic (classical) probability transition matrix. Moreover since $ \mathcal{T}$ is both unital and reversible we have that $P_{ij}= P_{ji}$. 
% AM removed and $\sum_j P_{ij} = 1$
So
\beas
\frac{1}{d} \sum_{i,j} \; g_i^a g_j^b \; \bra{i} {\cal L}(\proj{j}) \ket{i} &=& \frac{c_0}{d} \sum_{i,j} \; g_i^a g_j^b \; \left(\delta_{ij} - P_{ij}\right)\\ &=& \frac{c_0}{2d} \sum_{i,j} \left( g_i^a g_i^b P_{ij} + g_j^a g_j^b P_{ij} - g_i^a g_j^b P_{ij} - g_j^a g_i^b P_{ij}\right),
\eeas
where the second equality follows from $\sum_j P_{ij} = 1$, $P_{ij}= P_{ji}$. We can write this as
\[
{\cal E}_\mathcal{L}(g^a,g^b) =  \frac{c_0}{2d} \sum_{i,j} P_{ij} \left(g_i^a - g_j^a \right) \left(g_i^b - g_j^b\right).
\]

The claimed inequality then follows from a two-point inequality and then taking an average with respect to the measure $\frac{1}{d} P_{ij}$. It is shown in~\cite[proofs of Lemma 2.4 and Theorem 2.1]{mossel13a} that
\[
 pp'\left(x^\frac{1}{p} - y^\frac{1}{p}\right)\left(x^\frac{1}{p'} - y^\frac{1}{p'}\right) \leq  qq'\left(x^\frac{1}{q} - y^\frac{1}{q}\right)\left(x^\frac{1}{q'} - y^\frac{1}{q'}\right)
\]
for all $x,y \in \R_+$. This implies in particular for the eigenvalues $g_i \geq 0$ of $g$  that
\[
\frac{c_0 pp'}{2d} \sum_{i,j} P_{ij} \left(g_i^\frac{1}{p} - g_j^\frac{1}{p}\right)\left(g_i^\frac{1}{p'} - g_j^\frac{1}{p'}\right) \leq \frac{c_0 qq'}{2d} \sum_{i,j} P_{ij} \left(g_i^\frac{1}{q} - g_j^\frac{1}{q}\right)\left(g_i^\frac{1}{q'} - g_j^\frac{1}{q'}\right).
\]
Reversing the steps now yields the inequality as stated in the lemma.
\end{proof}

{\bf Remark}: The case $q=1$ can be obtained in terms of the appropriate limit.

By choosing suitable parameters ($p = p' = 2$) in Lemma \ref{lem:qstroock}, we obtain a generalization of a lemma of Gross~\cite{gross75} to the quantum setting, for $p<1$. A quantum version of this lemma was previously proven by King~\cite{king14} for $1 \le p \le 2$ in the case of unital qubit channels.

\begin{cor}[Quantum Gross's lemma for {$p \in (0,2]$}]
\label{cor:grossgen}
Let $\mathcal{T}$ be a unital, reversible and primitive quantum channel. Let ${\cal L} = c_0\left({\bf id} - \mathcal{T}\right)$, with $c_0 \in \R_+$, denote the generator of a reversible quantum semigroup. Then, for any $p \in (0,2]$,
\[ \mathcal{E}_\mathcal{L}(f^{p/2},f^{p/2}) \le \frac{p^2}{4(p-1)} \mathcal{E}_\mathcal{L}(f^{p-1},f). \]
\end{cor}

This corollary now lets us prove logarithmic Sobolev inequalities for $p<1$, given a logarithmic Sobolev inequality for $p=2$. This condition can be understood as an extension of $\mathbb{L}_p$ regularity, but restricted to unital channels \cite{kastoryano13,olkiewicz99}. We say that a semigroup generated by $\mathcal{L}$ satisfies a 2-log-Sobolev inequality with constant $\alpha$ if
\begin{equation} \operatorname{Ent}(f^2) \le \alpha \mathcal{E}_\mathcal{L}(f,f). \label{eqn:LSI}\end{equation}

{\bf Remark}: Primitivity of the semigroup is required, as otherwise the log-Sobolev inequality is trivial and the map is not hypercontractive. Indeed, primitivity guarantees that the stationary state of the semigroup is unique and has full rank. For unital semigroups, the maximally mixed state is stationary, hence full rank. To have a non-trivial log-Sobolev inequality we also have to ensure uniqueness of the fixed point. To see this note that for all reversible channels one has that $\alpha^{-1} \leq \lambda$, where $\lambda$ is the spectral gap of ${\cal L}$, cf.~\cite{olkiewicz99, kastoryano13}. For unital channels that are non-primitive we have that $\lambda = 0$ and hence the constant $\alpha$ diverges.

\begin{lem}[$p$-log-Sobolev inequalities]
\label{lem:plogsob}
Let $\mathcal{T}$ be a unital, reversible and primitive quantum channel. Let ${\cal L} = c_0\left({\bf id} - \mathcal{T}\right)$, with $c_0 \in \R_+$, denote the generator of a reversible quantum semigroup satisfying a 2-log-Sobolev inequality with constant $\alpha$. For any $p \in (0,2]$,
\[ \Ent(f^p) \le \frac{\alpha p^2}{4(p-1)} \mathcal{E}_\mathcal{L}(f^{p-1},f). \]
\end{lem}

\begin{proof}
Applying the 2-log-Sobolev inequality to $f^{p/2}$ and then Corollary \ref{cor:grossgen} to $g = f^p$, we have
\[ \Ent(f^p) \le \alpha \mathcal{E}_\mathcal{L}(f^{p/2},f^{p/2}) \le \frac{\alpha p^2}{4(p-1)} \mathcal{E}_\mathcal{L}(f^{p-1},f)\]
as required.
\end{proof}

The next lemma we will need is a technical claim regarding norm derivatives.

\begin{lem}[Norm derivative]
Let $\mathcal{T}_t$ be a Markovian family of quantum channels. Let $t:\R \rightarrow \R$ be a differentiable function of $p$ to be defined, and set $f_{t(p)} = \mathcal{T}_{t(p)}(f)$, with $f \geq 0$. Then
\label{lem:derivative}
\[ \frac{d}{dp} \ln \|\mathcal{T}_{t(p)}(f) \|_p = \frac{1}{p^2 \tau(f^p_{t(p)})}\left(\operatorname{Ent}(f_{t(p)}^p) -p^2 t'(p) \mathcal{E}_\mathcal{L}(f^{p-1}_{t(p)},f_{t(p)})\right). \]
\end{lem}

A variant of this lemma was proven in~\cite{olkiewicz99}. Here, we only consider the case where the reference state in the norm definition is the maximally mixed state. For this special case, the proof can be simplified.

\proof{Considering the function $p \mapsto \| f_{t(p)} \|_p = \tau(f^p_{t(p)})^{1/p}$ we can directly compute the derivative
\be\label{full-der}
	\frac{d}{dp} \ln \| f_{t(p)} \|_p = \frac{1}{p \| f_{t(p)} \|_p^p}\left( \frac{d}{dp} \| f_{t(p)} \|_p^p - \| f_{t(p)} \|_p^p \ln \| f_{t(p)} \|_p\right).
\ee
We now only need to evaluate the derivative $\frac{d}{dp} \|f_{t(p)}\|_p^p$. Note that we can write for any diagonalizable operator $f$ and any holomorphic function $h(z)$ that
\[
h(f) = \frac{1}{2\pi i} \int_{\partial \Delta} h(z) \frac{1}{z - f} dz, \;\; \mbox{where} \;\; \Delta \subset \C \:\; \mbox{so that} \;\; \mbox{spec}(f) \subset \Delta,
\]
by the Cauchy integral formula. In particular for  $h(z) = z^p$ and $f_{t(p)}$ positive definite, we can choose $\Delta$ to be supported entirely on the right half of the complex plane. When
$f_{t(p)}$ is positive semi-definite we can use a standard continuum argument by first perturbing $f$ so that it is positive definite and then taking the appropriate limit in the end. The derivative can now be evaluated as
\be\label{derivative}
	\frac{d}{dp} f_{t(p)}^p  = \frac{1}{2 \pi i} \int_{\partial \Delta} \left(z^p \ln z \frac{1}{z - f_{t(p)}} + z^p t'(p) \frac{d}{dt} \frac{1}{z - f_{t}} \right)dz.
\ee
If we choose the branch cut of the complex logarithm to be supported on the negative real axis the function $z^p \ln z$ is holomorphic on $\Delta$. Moreover, we can expand the matrix fraction by the standard formula as
\[
 \frac{1}{z - f_{t+ \epsilon}}  =  \frac{1}{z - f_{t}}  -  \epsilon \frac{1}{z - f_{t}}  {\cal L}(f_t)  \frac{1}{z - f_{t}} + {\cal O}(\epsilon^2),
\]
up to second order in $\epsilon >0$, from which we see that
\[
	 \frac{d}{dt} \frac{1}{z - f_{t}}= -  \frac{1}{z - f_{t}}  {\cal L}(f_t)  \frac{1}{z - f_{t}}.
\]
We now evaluate the full norm derivative. First applying Cauchy's integral formula for $z^p\ln z$, we then take the normalized trace $\tau$ on both sides of (\ref{derivative}). Using the cyclicity of the trace for the second summand we obtain
\be \label{norm-derivative}
\frac{d}{dp} \|f_{t(p)}\|_p^p  = \tau(f^p_{t(p)} \ln f_{t(p)})  - \frac{t'(p)}{2\pi i} \int_{\partial \Delta} \tau \left(\frac{z^p}{\left(z - f_{t(p)}\right)^2} {\cal L}(f_{t(p)})\right) dz.
\ee
Note that the second integral can be evaluated by considering the residuum at the second order pole so that
\[
 \frac{1}{2\pi i} \int_{\partial \Delta} \frac{z^p}{\left(z - f_{t(p)}\right)^2} dz = \mbox{Res}\left(\frac{z^p}{\left(z - f_{t(p)}\right)^2}\right) = pf^{p-1}_{t(p)}.
\]
Hence, if we now insert the derivative (\ref{norm-derivative}) into the full expression (\ref{full-der}), we obtain
\beas
\frac{d}{dp} \ln \| f_{t(p)} \|_p  &=&	\frac{1}{p^2 \|f_{t(p)}\|_p^p }\left(  \tau(f^p_{t(p)} \ln f^p_{t(p)}) -  \| f_{t(p)} \|_p^p \ln \| f_{t(p)} \|^p_p \right. \\\nonumber
					      &&       \left.  -p^2t'(p)\tau\left(f^{p-1}_{t(p)} {\cal L}(f_{t(p)})\right) \right),
\eeas
completing the proof. \qed}

Combining all these ingredients allows us to make a general statement about when 2-log-Sobolev inequalities can be lifted to $p$-log-Sobolev inequalities, and thence hypercontractive inequalities, for $0 \le p \le 2$.

\begin{thm}
\label{thm:logsob2p}
Let $\mathcal{T}$ be a unital, reversible and primitive quantum channel. Let ${\cal L} = c_0\left({\bf id} - \mathcal{T}\right)$, with $c_0 \in \R_+$, denote the generator of a reversible quantum semigroup $\mathcal{T}_t$ satisfying a 2-log-Sobolev inequality with constant $\alpha$.

Let $f \in \mathcal{M}_d$ be a positive semidefinite operator and let $-\infty < q \le p \le 1$ (if $p < 0$, assume that $f$ is positive definite). Then, for all $t$ such that $t \ge \frac{\alpha}{4} \ln ((1-q)/(1-p))$,
\[ \| \mathcal{T}_t(f) \|_q \ge \|f\|_p. \]
\end{thm}

\begin{proof}
The proof is essentially the same as the classical proof in~\cite{mossel13a}. We split into three cases. First assume that $p > q \ge 0$. Consider a function
\[ t(q) = \frac{\alpha}{4} \ln \frac{1-q}{1-p} \]
defined on $(0,p]$. Then
\[ t(p) = 0,\;\;\;\; q^2 t'(q) = \frac{\alpha q^2}{4(q-1)}. \]
Now consider the map $q \mapsto \|\mathcal{T}_t(f)\|_q$. By Lemma \ref{lem:derivative},
\[ \frac{d}{dq} \ln \|\mathcal{T}_t(f) \|_q = \frac{\operatorname{Ent}(f_{t(q)}^q) - \frac{\alpha q^2}{4(q-1)} \mathcal{E}_{\mathcal{L}}(f^{q-1}_{t(q)},f_{t(q)})}{q^2 \tau(f^q_{t(q)})} \le 0, \]
where the inequality follows from Lemma \ref{lem:plogsob}. At the right-hand end of the interval the map evaluates to $\|f\|_p$. Therefore,
\[ \|\mathcal{T}_t(f)\|_q \ge \|f\|_p, \]
or in other words
\[ \|\mathcal{T}_{ \frac{\alpha}{4} \ln ((1-q)/(1-p))} (f)\|_q \ge \|f\|_p. \]
The generalization to $t \ge \frac{\alpha}{4} \ln((1-q)/(1-p))$ follows from monotonicity of norms by setting $t = \frac{\alpha}{4}\ln ((1-r)/(1-p))$ for some $r \le q$.

Second, assume that $q < 0 \le p$. Here we have
\[ \| \mathcal{T}_t(f) \|_q = \| \mathcal{T}_{t - \frac{\alpha}{4} \ln (1/(1-p))}(\mathcal{T}_{\frac{\alpha}{4} \ln (1/(1-p))}(f))\|_q \ge \|\mathcal{T}_{\frac{\alpha}{4} \ln (1/(1-p))}(f)\|_0 \ge \|f\|_p \]
using the first case and that $t \ge \frac{\alpha}{4} \ln ((1-q)/(1-p))$ implies $t - \frac{\alpha}{4} \ln (1/(1-p)) \ge \frac{\alpha}{4} \ln (1-q)$.

Third, the case $q < p<0$ is proven by a duality argument using Lemma \ref{lem:varchar}:
\beas
\| \mathcal{T}_t(f) \|_q &=& \inf \{ \tau(g \mathcal{T}_t(f)): g > 0, \|g\|_{q'} \ge 1 \}\\
&=& \inf \{ \tau(f \mathcal{T}_t(g)): g > 0, \|g\|_{q'} \ge 1 \}\\
&\ge& \inf \{ \tau(f h): h > 0, \|h\|_{p'} \ge 1 \}\\
&=& \|f\|_p,
\eeas
where we as usual define $p'$, $q'$ by $1/p + 1/p' = 1/q+1/q'=1$. The second equality follows from reversibility of $\mathcal{T}_t$, while the inequality holds because $t \ge \frac{\alpha}{4} \ln ((1-q)/(1-p)) = \frac{\alpha}{4} \ln ((1-p')/(1-q'))$, so for $h = \mathcal{T}_t(g)$, $\|h\|_{p'} \ge \|g\|_{q'}$. This completes the proof.
\end{proof}

The special case of the qubit depolarizing channel (Theorem \ref{thm:reversequantumdep}) is now simply a corollary.

\begin{repthm}{thm:reversequantumdep}
Let $f \in \mathcal{M}_{2^n}$ be a positive semidefinite operator and let $-\infty < q \le p \le 1$ (if $p < 0$, assume that $f$ is positive definite). Then, for all $\gamma$ such that $0 \le \gamma \le \sqrt{(1-p)/(1-q)}$,
\[ \| \mathcal{D}_\gamma^{\otimes n}(f) \|_q \ge \|f\|_p. \]
\end{repthm}

\begin{proof}
The assumptions of Lemmas \ref{lem:qstroock} to \ref{lem:plogsob} are all met by $n$ copies, $\mathcal{D}_\gamma^{\otimes n}$, of the single qubit depolarizing channel $\mathcal{D}_\gamma$. This can be seen by writing ${\cal L}_k(f) =  f  - \frac{1}{4}\sum_{\alpha = 0}^3 \sigma^\alpha_k f \sigma^\alpha_k$ for the generator of the semigroup acting on the $k$'th qubit, where $\sigma^\alpha$ are the Pauli matrices. Then the overall generator ${\cal L} = \sum_{k=1}^n {\cal L}_k$. In addition, it was shown in~\cite{kastoryano13} that the tensor product of qubit depolarizing channels satisfies a 2-log-Sobolev inequality with constant $\alpha=2$: $\operatorname{Ent}(f^2) \le 2 \mathcal{E}_\mathcal{L}(f,f)$. We can therefore apply Theorem \ref{thm:logsob2p} to the channel $\mathcal{D}_\gamma^{\otimes n}$ to obtain the claimed result.
\end{proof}

% ------------------------------------------------------------------------------

\section{Reverse H\"older and rapid mixing}

We can now collect some corollaries of hypercontractivity for tensor products of qubit depolarizing channels. We first observe a strengthened reverse H\"older inequality.

\begin{cor}
\label{cor:strongrevholder}
Let $f,g \in \mathcal{M}_{2^n}$ be positive semidefinite and let $-\infty < q, p \le 1$ (if $p < 0$, $f$ must be positive definite; if $q < 0$, $g$ must be positive definite). Then, for all $\gamma$ such that $0 \le \gamma \le \sqrt{(1-p)(1-q)}$,
\[ \tau(f \mathcal{D}_\gamma^{\otimes n}(g)) \ge \|f\|_p \|g\|_q. \]
\end{cor}

\begin{proof}
First observe that it is sufficient to prove the claim for $\gamma = \sqrt{(1-p)(1-q)}$. Otherwise, set $\gamma = \sqrt{(1-p)(1-r)}$ for some $r \ge q$ and observe
\[ \tau(f \mathcal{D}_\gamma^{\otimes n}(g)) \ge \|f\|_p \|g\|_r \ge \|f\|_p \|g\|_q \]
by monotonicity. Assuming $\gamma = \sqrt{(1-p)(1-q)}$, let $p'$ satisfy $1/p + 1/p' = 1$ and use the reverse H\"older inequality to obtain
\[ \tau(f \mathcal{D}_\gamma^{\otimes n}(g)) \ge \|f\|_p \|\mathcal{D}_\gamma^{\otimes n}(g)\|_{p'}. \]
As $1/(1-p') = 1-p$, we have $\gamma = \sqrt{(1-p)(1-q)} = \sqrt{(1-q)/(1-p')}$. By the reverse hypercontractive inequality (Theorem \ref{thm:reversequantumdep}),
\[ \|\mathcal{D}_\gamma^{\otimes n}(g)\|_{p'} \ge \|g\|_q \]
as required to complete the proof.
\end{proof}

Following~\cite{mossel06}, we now use Corollary \ref{cor:strongrevholder} to prove the following theorem:

\begin{thm}
\label{thm:mixing}
Let $S$ be a subspace of $(\C^2)^{\otimes n}$ with corresponding projector $\Pi_S$, such that $\dim S = \exp(-s^2/2)2^n$ for some $s \ge 0$. Set $\rho_S = \Pi_S / \dim S$. Let $M \in \mathcal{M}_{2^n}$ satisfy $0 \le M \le I$, $\tau(M) = \exp(-t^2/2)$ for some $t \ge 0$. %Let $\ket{\psi} \in (\C^2)^{\otimes n}$ be a random state from $S$. %, and set $\rho = \mathcal{D}_\gamma^{\otimes n}(\proj{\psi})$.
Then
\[ \tr(M \mathcal{D}_\gamma^{\otimes n}(\rho_S)) \ge \exp\left(-\frac{1}{2} \left(\frac{s^2+2\gamma s t + t^2}{1-\gamma^2} - s^2\right) \right). \]
\end{thm}

\begin{proof}
The proof is effectively an immediate consequence of Corollary \ref{cor:strongrevholder}. First, we have
\[ \tr(M \mathcal{D}_\gamma^{\otimes n}(\rho_S)) = \frac{2^n}{\dim S} \tau(M \mathcal{D}_\gamma^{\otimes n}(\Pi_S)) = \exp(s^2/2) \tau(M \mathcal{D}_\gamma^{\otimes n}(\Pi_S)). \]
By Corollary \ref{cor:strongrevholder}, for any $p,q < 1$ such that $(1-p)(1-q) = \gamma^2$,
\[ \tau(M \mathcal{D}_\gamma^{\otimes n}(\Pi_S)) \ge \|\Pi_S\|_p \|M\|_q \ge \exp(-s^2/(2p)) \exp(-t^2/(2q)), \]
where we use $0 \le M \le I$ in the second inequality. As in~\cite[Theorem 3.4]{mossel06} we maximise the right-hand side by picking
\[ p = \frac{1-\gamma^2}{1+\gamma(t/s)},\;\;\;\; q = \frac{t}{s} \frac{1-\gamma^2}{\gamma + (t/s)}, \]
which yields
\[  \tr(M \mathcal{D}_\gamma^{\otimes n}(\rho_S)) \ge \exp\left(-\frac{1}{2} \left( \frac{s^2+2\gamma s t + t^2}{1-\gamma^2} - s^2\right) \right). \]
\end{proof}

Note that Theorem \ref{thm:mixing} holds for any quantum channels satisfying reverse hypercontractivity, not just tensor powers of the depolarizing channel. By fixing parameters in Theorem \ref{thm:mixing}, we obtain the following special case:

\begin{cor}
\label{cor:mixing}
Let $S$ be a subspace of $(\C^2)^{\otimes n}$ with corresponding projector $\Pi_S$, such that $\dim S = \sigma 2^n$ for some $\sigma$. Let $M \in \mathcal{M}_{2^n}$ satisfy $0 \le M \le I$, $\tau(M) = \sigma^\alpha$ for some $\alpha \ge 0$. Set $\rho_S = \Pi_S / \dim S$. Then
\[ \tr(M \mathcal{D}_\gamma^{\otimes n}(\rho_S)) \ge \sigma^{(\sqrt{\alpha}+\gamma)^2/(1-\gamma^2)}. \]
In the special case $\alpha = 1$, this is at least $\sigma^{(1+\gamma)/(1-\gamma)}$.
\end{cor}

Let us examine what this corollary is saying. Write $\rho_\gamma = \mathcal{D}_\gamma^{\otimes n}(\rho_S)$. If $\gamma = 0$, $\rho_\gamma$ is maximally mixed, so $\tr(M \rho_\gamma) = \sigma^\alpha$. Corollary \ref{cor:mixing} matches this. More generally, imagine $\sigma$, $\alpha$ and $\gamma$ are fixed constants. Then Corollary \ref{cor:mixing} states that $\tr(M \rho_\gamma)$ is also lower-bounded by a constant, independent of dimension. This is a kind of ``rapid mixing'' result for tensor products of qubit depolarizing channels: if we start with a random state picked from $S$ and apply depolarizing noise, the resulting state is quite likely to be accepted by the measurement operator $M$, even if, for example, $M$ is the projector onto a subspace orthogonal to $S$.

% ------------------------------------------------------------------------------

\section{Non-interactive correlation distillation}

We now apply Theorem \ref{thm:reversequantumdep} to a quantum generalization of the problem of non-interactive correlation distillation~\cite{mossel05,mossel06,mossel13a}. Classically, this problem can be defined in terms of a game involving $k$ players. A string of $n$ uniformly random bits is produced and a copy of the $n$-bit string is distributed to each of the players. Each bit in each of the copies is subject to independent noise, being flipped with probability $\epsilon = (1-\gamma)/2$. Each of the players applies the same boolean function $f:\{0,1\}^n \rightarrow \{0,1\}$ to their bit-string. Their aim is that the $k$ output values are all equal to some $y \in \{0,1\}$ such that $y$ is uniformly random.

The following bound on the success probability was shown by Mossel et al.~\cite{mossel06} using reverse hypercontractivity:

\begin{thm}[Mossel et al.~\cite{mossel06}]
\label{thm:nicdlowerclass}
For any function $f$, and any noise rate $0 \le \gamma \le 1$, the probability that all players output 0 (or 1) is at most
\[ O\left( \left( \frac{e^{c \sqrt{\ln k}} }{k} \right)^{1/\gamma^2 - 1} \right) \]
for some universal constant $c$.
\end{thm}

The success probability is thus exponentially small in $k$ for any constant noise rate $\gamma$. Theorem \ref{thm:nicdlowerclass} is close to optimal: if $f$ is the majority function, for large enough $n$ the probability that all players output the same value is $\Omega(k^{1-1/\gamma^2})$~\cite{mossel06}.

We now observe that Theorem \ref{thm:reversequantumdep} allows us to prove a related and more general result in the quantum setting. In the quantum generalization, the players first fix an orthonormal basis $\mathcal{B}$ for $(\C^2)^{\otimes n}$. A state $\ket{\psi} \in (\C^2)^{\otimes n}$ is chosen uniformly at random from $\mathcal{B}$ by a referee. Each of the $k$ players receives a copy of $\ket{\psi}$, with independent depolarizing noise with parameter $\gamma$ applied to each qubit of each copy of $\ket{\psi}$. Each player then applies a ``balanced'' two-outcome measurement $\{M,I-M\}$ to their state, where $0 \le M \le I$ is a positive semidefinite operator such that $\tau(M) = \frac{1}{2}$. It is natural to demand this notion of balance, as even in the noiseless case ($\gamma=1$) this is necessary in order to obtain equiprobable measurement outcomes, and the players do not necessarily know the noise parameter. As in the classical setting, the goal is for the measurement outcomes to be random and perfectly correlated. That is, either every player should receive the measurement outcome corresponding to $M$ or every player should receive the outcome corresponding to $I-M$, with equal probability of each.

This is a generalization of the classical framework: if the players choose $\mathcal{B}$ to be a product basis for $(\C^2)^{\otimes n}$, the game behaves equivalently to classical non-interactive correlation distillation. In principle, it could be possible for the players to do better by letting $\mathcal{B}$ be a basis of entangled states. However, we show here that an equivalent bound to Theorem \ref{thm:nicdlowerclass} can be proven for this more general task.

\begin{thm}
For any balanced measurement $\{M, I-M\}$, and any noise rate $0 \le \gamma \le 1$, the probability that all players receive outcome $M$ (or $I-M$) is at most
\[ O\left( \left( \frac{e^{c \sqrt{\ln k}} }{k} \right)^{1/\gamma^2 - 1} \right) \]
for some universal constant $c$.
\end{thm}

\begin{proof}
The proof strategy of~\cite[Theorem 3.1]{mossel06} for the classical setting goes through with only minor changes. Assume that the expected probability that all the players output 0 is at least $2\delta$, i.e.
\[ \E_{\ket{\psi}}[(\tr M \mathcal{D}_\gamma^{\otimes n}(\proj{\psi}) )^k] \ge 2\delta, \]
where the expectation is over the uniformly random choice of $\ket{\psi}$ from $\mathcal{B}$. Let $S = \spann\{ \ket{\psi}: (\tr M \mathcal{D}_\gamma^{\otimes n}(\proj{\psi}) )^k \ge \delta \}$ and set $\sigma = (\dim S)/2^n$. Using a similar argument to the proof of Markov's inequality, we have $\Pr_{\ket{\psi}}[(\tr M \mathcal{D}_\gamma^{\otimes n}(\proj{\psi}) )^k \ge \delta] \ge \delta$, and hence $\sigma \ge \delta$. By the definition of $S$, for any $\ket{\phi} \in S$
\be \label{eq:ineq1} \tr[(I-M) \mathcal{D}_\gamma^{\otimes n}(\proj{\phi})]  \le 1 - \delta^{1/k}, \ee
so if we write $\rho_S = \Pi_S / \dim S$, $\tr[(I-M) \mathcal{D}_\gamma^{\otimes n}(\rho_S)] \le 1-\delta^{1/k}$.

But we can also apply Corollary \ref{cor:mixing} to $S$ and $I-M$, where $\tau(M) = \frac{1}{2}$ and hence $\alpha = 1/(\log_2 1/\sigma)$. This implies that
\[ \tr[((I-M) \mathcal{D}_\gamma^{\otimes n}(\rho_S)] \ge \sigma^{(1/\sqrt{\log_2 1/\sigma}+\gamma)^2/(1-\gamma^2)}. \]
For any $0 \le \gamma \le 1$, the right-hand side is an increasing function of $\delta$, and we therefore have the bound
\be \label{eq:ineq2} \tr[(I-M) \mathcal{D}_\gamma^{\otimes n}(\rho_S)] \ge \delta^{(1/\sqrt{\log_2 1/\delta}+\gamma)^2/(1-\gamma^2)} \ge \delta^{(1/\sqrt{\ln 1/\delta}+\gamma)^2/(1-\gamma^2)}. \ee
The rest of the proof follows by combining inequalities (\ref{eq:ineq1}) and (\ref{eq:ineq2}) to upper-bound $\delta$, exactly as in~\cite{mossel06}. Write $\nu = 1/\gamma^2 - 1$. We will show that, if $\delta \ge (e^{c \sqrt{\ln k}}/k)^\nu$, for a sufficiently large universal constant $c$, then $\delta^{1/k} + \delta^{(1/\sqrt{\ln 1/\delta}+\gamma)^2/(1-\gamma^2)} > 1$ for large $k$, violating the combination of (\ref{eq:ineq1}) and (\ref{eq:ineq2}). First we have
\[ \delta^{1/k} > \left(\frac{1}{k} \right)^{\nu/k} = e^{-(\nu \ln k)/k} > 1 - \frac{\nu \ln k}{k}. \]
We can also expand
\[ \delta^{(1/\sqrt{\ln 1/\delta}+\gamma)^2/(1-\gamma^2)} = \delta^{(\ln^{-1} 1/\delta)/(1-\gamma^2)} \delta^{2\gamma(\ln^{-1/2} 1/\delta)/(1-\gamma^2)} \delta^{\gamma^2/(1-\gamma^2)} \]
and bound these three parts as
\beas
\delta^{(\ln^{-1} 1/\delta)/(1-\gamma^2)} &=& e^{-1/(1-\gamma^2)},\\
\delta^{2\gamma(\ln^{-1/2} 1/\delta)/(1-\gamma^2)} &=& e^{-2\gamma\sqrt{\ln 1/\delta}/(1-\gamma^2)} \ge e^{-2\gamma\sqrt{\nu \ln k}/(1-\gamma^2)},\\
\delta^{\gamma^2/(1-\gamma^2)} &=& \delta^{1/\nu} \ge \frac{e^{c \sqrt{\ln k}}}{k}.\\
\eeas
For sufficiently large $c$ (where $c$ depends on $\gamma$, but not on $k$), and large enough $k$, the product of these three terms is larger than $(\nu \ln k) / k$. This implies the desired contradiction and completes the proof.
\end{proof}

We remark that, classically, reverse hypercontractivity has also been applied to noninteractive distillation of correlations in a more general setting~\cite{kamath12}. Delgosha and Beigi~\cite{delgosha14} have shown that ``standard'' quantum hypercontractivity can be used to put limits on correlation distillation, via a quantum generalization of the notion of the hypercontractivity ribbon of Ahlswede and G\'acs~\cite{ahlswede76}. It seems likely that the quantum reverse hypercontractive inequality could be applied in a similar way to improve these results, but we do not pursue this further here.

% ------------------------------------------------------------------------------

\subsection*{Acknowledgements}

This work was initiated at the BIRS workshop 15w5098, ``Hypercontractivity and Log Sobolev Inequalities in Quantum Information Theory''. We would like to thank BIRS and the Banff Centre for their hospitality. We would also like to thank Mark Wilde for pointing out reference~\cite{tomamichel14}, and an anonymous referee for helpful comments. MK was supported by the Carlsbergfond and the Villum foundation. AM was supported by the UK EPSRC under Early Career Fellowship EP/L021005/1.
K.T. was supported by the Institute for Quantum Information and Matter, a NSF Physics Frontiers Center with support of the Gordon and Betty Moore Foundation (Grants No. PHY-0803371 and PHY-1125565). TSC is supported by the Royal Society.

% ------------------------------------------------------------------------------

%\bibliographystyle{plain}
%\bibliography{reverse}

\end{document}